\documentclass[letterpaper, 10 pt, conference]{ieeeconf}  % Comment this line out if you need a4paper

\IEEEoverridecommandlockouts                              % This command is only needed if 
\usepackage{amsmath} % assumes amsmath package installed
\usepackage{amssymb}  % assumes amsmath package installed
\usepackage{xcolor}
\usepackage{graphicx}
\usepackage{booktabs}
\usepackage{multirow}
\usepackage{algorithm}
\usepackage{algpseudocode}
\usepackage{upgreek}

\newtheorem{theorem}{Theorem}
\newtheorem{proposition}{Proposition}
\newtheorem{remark}{Remark}

\renewenvironment{proof}[1][Proof]{%
  \par\noindent\textit{#1.}\enspace\ignorespaces
}{%
  \unskip\nobreak\hfill\QEDopen\par
}

\title{\LARGE \bf
    DAOCP: a dual active set solver for optimal control problems
}

\author{Alberto Zaupa$^{*}$, Samuel Erickson$^{*}$ and Mikael Johansson$^{*}$% <-this % stops a space
\thanks{$^*$ The authors are with the Department of Decision and Control Systems at KTH Royal Institute of Technology, Stockholm, Sweden. Email addresses \texttt{\{zaupa, samuelea, mikaelj\}@kth.se}}%
}

\begin{document}

\maketitle
\thispagestyle{empty}
\pagestyle{empty}

%%%%%%%%%%%%%%%%%%%%%%%%%%%%%%%%%%%%%%%%%%%%%%%%%%%%%%%%%%%%%%%%%%%%%%%%%%%%%%%%
\begin{abstract}
We present DAOCP, a dual active set solver for linear quadratic optimal control problems with stage-wise equality and inequality constraints. Active set methods are leading Model Predictive Control benchmarks for full-body robotics, but existing solvers operate on dense QPs, while typical problem dimensions favor methods that exploit the optimal control structure. DAOCP combines the warm-starting capabilities of active set algorithms with the better computational scaling of structure exploiting solvers, by relying on a generalization of the relationship between the Riccati recursion and the Cholesky factorization of the condensed Hessian. This result enables dual active set iterations to operate directly on the original optimal control problem through recursive computations, without explicitly forming the condensed quadratic program. On robotics benchmarks, DAOCP is the fastest of four solvers in four of five scenarios, cutting average solve time on a 58-state Atlas model by 9$\times$ relative to state of the art solvers DAQP and HPIPM.  

\end{abstract}

%%%%%%%%%%%%%%%%%%%%%%%%%%%%%%%%%%%%%%%%%%%%%%%%%%%%%%%%%%%%%%%%%%%%%%%%%%%%%%%%
\section{Introduction}

Model Predictive Control (MPC) is a core methodology in robotics and 
increasingly used to control complex robotic systems with a large number of
degrees of freedom~\cite{DC2018,MN2017,D2021}. Full-body MPC for
humanoid and legged robots requires solving optimization problems with large
state and control dimensions over moderate to long prediction horizons, within
tight real-time budgets. This dimensionality regime is particularly favorable
for solvers that directly exploit the Optimal Control Problem (OCP) structure,
rather than transforming the problem into a dense Quadratic Program (QP). At
the same time, Active Set (AS) methods have a natural algorithmic advantage
in MPC, due to their ability to effectively exploit warm starts  on closely
related problems. This is reflected by the leading performance of AS methods
like DAQP~\cite{DA2022} on robotics-inspired benchmarks~\cite{qpbenchmark}.

Structure exploiting solvers for MPC have a long history, particularly among 
Interior Point methods~\cite{YW2010, GF2020}. Several authors
have proposed AS algorithms for
OCPs~\cite{DA2006,JB2024,AS2011,IN2013}, but these have not led to the
development of state of the art MPC solvers. Modern Active Set solvers still rely on the condensed MPC formulation~\cite{DA2022,JF2014}, and a competitive implementation designed specifically for OCPs remains missing.

This paper presents DAOCP~\footnote{The solver code is available at https://github.com/AlbertoZaupa/daocp}, a dual AS solver aiming to close this gap. The
solver adopts the high-level algorithm of DAQP and Quadprog~\cite{G1983},
but exploits the OCP structure and avoids condensing. To do so, we generalize
a correspondence between the Riccati recursion and the Cholesky factorization
of the condensed Hessian from~\cite{GF2013_TRCHOL} to problems with stage-wise
equality constraints.
On robotics benchmarks, DAOCP is the fastest of four solvers in four of five
scenarios, cutting average solve time on a 58-state Atlas model by 9$\times$
relative to DAQP and HPIPM. Figure~\ref{fig:atlas} shows the closed-loop
solve times on this benchmark.

\begin{figure}[!t]
    \centering
    \includegraphics[width=\columnwidth]{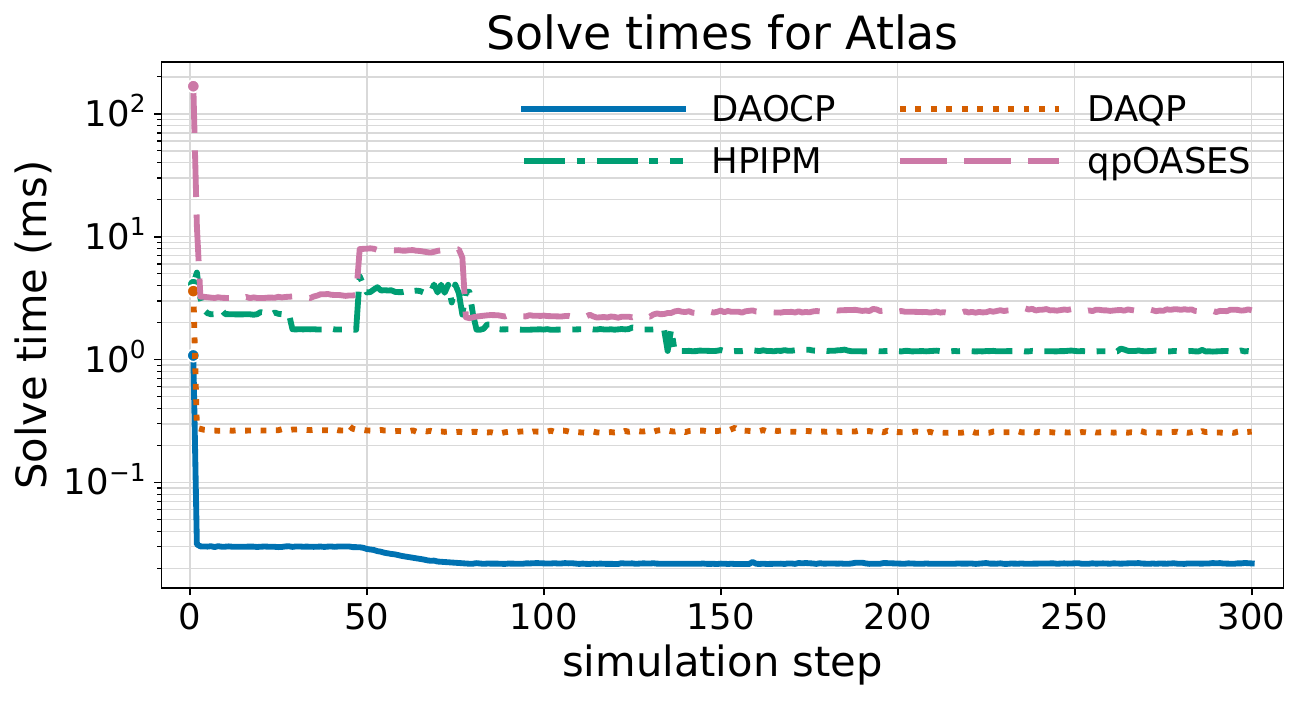}
    \caption{Solve times for the Atlas benchmark. Timings are averaged over 50 independent runs on a per-simulation step basis.}
    \label{fig:atlas}
\end{figure}

The remainder of the paper is organized as follows. In
Section~\ref{sec:background} we introduce the necessary background on MPC
solvers and AS methods in particular. In Section~\ref{sec:daocp} we present the
DAOCP algorithm, with specific focus on how to exploit OCP structure of the
primal problem. The solver is evaluated on a variety of robotics-inspired
benchmarks in Section~\ref{sec:numerical_experiments}, where we demonstrate
its performance advantages over other state of the art MPC solvers. Finally, we
conclude the paper in Section~\ref{sec:conclusion}.

\section{Background}
\label{sec:background}
\subsection{Solving MPC problems}
In Model Predictive Control, linear quadratic Optimal
Control Problems (OCPs) are solved in receding horizon
fashion, meaning that at every sampling time the controller
solves problems of the form
\begin{gather}
\label{prob:ocp_problem}
\begin{array}{ll} 
         \underset{{\{x_t\}, \{u_t\}}}{\mbox{minimize}} &\displaystyle \frac{1}{2} x_N^{\top} Q_N x_N + q_N^{\top} x_N + \\
        & \displaystyle\sum_{t=0}^{N-1} \frac{1}{2}\begin{bmatrix}
         u_t \\ x_t
        \end{bmatrix}^{\top} \begin{bmatrix}
         R_t & S_t \\
         S_t^{\top} & Q_t
        \end{bmatrix} \begin{bmatrix}
         u_t \\ x_t
        \end{bmatrix} + \begin{bmatrix}
         r_t \\ q_t
        \end{bmatrix}^{\top} \begin{bmatrix}
        u_t \\ x_t
        \end{bmatrix} \\
        \mbox{subject to} & \quad x_{t+1} = A_t x_t + B_t u_t + w_t \\
        &\quad C^x_t x_t + C^u_t u_t \leq c_t \\
        &\quad C_N^x x_N \leq c_N \\
        &\quad x_0 \text{ given.}
\end{array}
\end{gather}
where $t$-dependent constraints are applied for $t=0,\dots,N-1$. We assume that
$R_t \succ 0$, $Q_N \succeq 0$, and that the stage cost matrices are positive
semi-definite. 
We denote the objective of Problem~\eqref{prob:ocp_problem} by
$J_{\text{LQR}}$. Only the first optimal control variable is applied to the system,
and the rest of the solution is used as a warm start for the following sampling
time.

Problem~\eqref{prob:ocp_problem} can be solved with standard quadratic
programming solvers after formulating it as a quadratic program (QP) in
standard form
\begin{equation}
    \begin{array}{ll} 
        \underset{z}{\mbox{minimize}}  & \frac{1}{2} z^{\top} H z + g^{\top} z \\
        \mbox{subject to}              & \quad Gz \leq b\, .
    \end{array}
\end{equation}
When doing so we can choose to either retain state variables, as in the
\textit{extensive} formulation, or to eliminate them through the dynamical
relationship $x_{t+1} = A_t x_t + B_t u_t + w_t$, leading to the
\textit{condensed} formulation. In the former case, the
problem matrices $H$ and $G$ are sparse, while in the latter case they are dense.
Extensive form QPs can be solved with
sparse QP solvers like OSQP~\cite{BS2020} and Clarabel~\cite{PG2026}, while solvers like
DAQP~\cite{DA2022} and qpOASES~\cite{JF2014} perform the
best on condensed problems.

Alternatively, one may employ dedicated solvers for OCP-structured problems,
like HPIPM~\cite{GF2020} and QPALM-OCP~\cite{KL2024}.
These approaches solve the extensive formulation by
exploiting the specific sparsity pattern of the OCP. This is advantageous
because it allows using dense, rather than general sparse, linear algebra
techniques, which can more effectively utilize the computational resources of
modern hardware. For this reason OCP solvers consistently outperform sparse QP
solvers, especially for moderate to large
numbers of state and control variables. 

The solution approaches described so far apply to linear quadratic OCPs,
however they are often employed as core subroutines of solvers for nonlinear MPC and nonlinear
trajectory optimization. For example, QP solvers are the
computational backbone of the Sequential Quadratic Programming (SQP)
algorithms implemented in Acados~\cite{RV2022}. In an SQP
setting, dedicated OCP solvers have a particular structural advantage: because
intermediate QPs are naturally generated in OCP format, dense solvers pay the extra
overhead of condensing the problem at every iteration.

\subsection{When OCP structure matters}
\label{sec:when_ocp_matters}
The main numerical operations performed by optimization algorithms for QPs and OCPs can broadly be divided into two main classes: linear system factorizations and linear system solves. The former are always needed during solver setup, but may not play a critical role across solver iterations. Instead linear system solves, and related routines with similar floating point operations (flop) counts, dominate the per-iteration cost of AS and First Order methods, but are not particularly relevant for Interior Point algorithms.

OCP solvers can factorize the linear system by running a finite horizon Riccati recursion, which is then used to solve the linear system through a related LQR recursion which computes the optimal state-input sequence~\cite{GF2013}. 

The flop count of a classical Riccati recursion is $\mathcal{C}_f = N(\frac{7}{3}n_x^3 + 4n_x^2n_u + 2n_x n_u^2 + \frac{1}{3}n_u^3)$, while the cost of the LQR recursion is $\mathcal{C}_s = N(4 n_x^2 + 4 n_x n_u + 2n_u^2)$ (see~\cite{GF2013} for more details). Comparing with the corresponding operations generally performed by dense QP solvers, which often rely on the Cholesky decomposition to factorize the linear system, one can establish the following threshold values for $N$, as a function of $\alpha = n_x / n_u$
\begin{equation}
\label{eq:N_thresholds}
\begin{aligned}
N_f &= \sqrt{7\alpha^3 + 12 \alpha^2 + 6 \alpha + 1} \\
N_s &= 2\alpha^2 + 2\alpha + 1 \,.
\end{aligned}
\end{equation}
When $N \geq N_f$ Riccati recursions are cheaper than Cholesky factorizations, while $N \geq N_s$ favors LQR solves over standard linear system solves.

For robotics and other mechanical systems with approximately one input and two state variables per degree of freedom, $\alpha = 2$ is a representative
value, which gives $N_f \approx 10.8$ and $N_s = 13$. Given that typical values of the horizon exceed these thresholds, often by a significant margin~\cite{ED2022},~\cite{MN2017},~\cite{D2021}, robotic applications are a natural fit for OCP-structure exploiting solvers.

Raw flop counts however don't give a complete picture. In fact, Riccati and LQR recursions struggle to achieve high utilization of the computational resources when $n_x$ and $n_u$ are small, due to the lack of sufficient parallelism in the matrix operations. However this is not an issue in the context of whole-body control for robotics, where $n_x$ and $n_u$ often range between 20 and 50~\cite{JPS2021},~\cite{D2021}, and are therefore large enough to make the proposed flops comparison meaningful. 

\subsection{Active Set algorithms}
Active Set (AS) algorithms are particularly attractive
for MPC because they can exploit the similarity between QPs at consecutive
sampling times very effectively. These methods maintain a working set
of constraints treated as equalities, which is updated until the optimality
conditions are satisfied. When the changes in the optimal active set between
sampling times are incremental, warm starting can substantially reduce the number
of iterations~\cite{JF2014, DA2022}. This is a practical advantage over
Interior Point methods, which benefit notoriously little from warm starting. AS
methods can also be competitive without a warm start~\cite{DA2022}, contrary to
First Order methods~\cite{BS2020} which tend to be quite unreliable in this
setting, especially when high solution accuracy is required.

Classical AS methods can be classified as \textit{primal} or \textit{dual},
according to whether they preserve primal or dual feasibility across iterates.
An advantage of dual methods is that the origin is always a feasible initial
point, contrary to primal methods where a Phase I procedure is typically
needed. In part for this reason, dual active set methods often outperform their
primal counterparts~\cite{G1983},~\cite{GC2017}.

Established AS solvers such as qpOASES~\cite{JF2014} and DAQP~\cite{DA2022} are
commonly applied to condensed MPC problems. However, there are also previous
works on AS algorithms for OCPs. In~\cite{DA2006} the Riccati recursion is used
to solve subproblems generated by a dual AS algorithm for QP relaxations
arising in Mixed Integer MPC. Other works explore homotopy-based~\cite{JB2024}
and primal AS methods~\cite{AS2011} for OCP structured problems, respectively.
To the best of our knowledge, no public implementation of these algorithms is
available.

\section{DAOCP}
\label{sec:daocp}
In this section we present the DAOCP algorithm for solving problems of the
form~\eqref{prob:ocp_problem}, beginning with the high level active set
procedure. Here we address a slightly different version of
Problem~\eqref{prob:ocp_problem}, where we explicitly differentiate between
stage-wise equality and inequality constraints, in order to deal with equality
constraints more efficiently. More specifically, we are interested in solving
\begin{gather}
\label{prob:ocp_problem_eqcon}
\begin{array}{ll}
        \underset{\{x_t\}, \{u_t\}}{\mbox{minimize}} &\quad \displaystyle J_{\text{LQR}} \\
       \mbox{subject to} & x_{t+1} = A_t x_t + B_t u_t + w_t \\
       & C^x_t x_t + C^u_t u_t \leq c_t \\
       &D^x_t x_t + D^u_t u_t = d_t \\
       &C_N^x x_N \leq c_N \\
       &x_0 \text{ given,}
\end{array}
\end{gather}
where $x_t \in \mathbb{R}^{n_x}$ and $u_t \in \mathbb{R}^{n_u}$. We denote the
number of inequalities per stage by $m_t$, and the number of equality
constraints per stage by $\rho_t$ (excluding the dynamics), and define $m =
\sum_{t=0}^{N} m_t$, $\rho = \sum_{t=0}^{N-1} \rho_t$. Note that in
Problem~\eqref{prob:ocp_problem_eqcon}, we do not include a terminal equality
constraint of the form $D_N x_N = d_N$. This is because one can implement a procedure based on a similar argument to~\cite{LV2024}, that either detects infeasibility of the equality constraints for
the given $x_0$, or rewrites an OCP with arbitrary stage-wise and terminal
equality constraints in the form of Problem~\eqref{prob:ocp_problem_eqcon},
where the matrices $D_t^u$ are full row-rank at every stage. We describe the procedure in Appendix~\ref{app:eqcon_lqr}.

In order to introduce the dual AS algorithm, it is more convenient to rewrite
Problem~\eqref{prob:ocp_problem_eqcon} in its condensed form, obtained by
eliminating state variables. We use this form only to describe the algorithm,
while the implementation will exploit the original OCP structure without
explicitly condensing the problem. We express state variables as
\begin{equation}
\label{eq:condensing_equation}
\mathbf{x} = \mathcal{A} x_0 + \mathcal{B} \mathbf{u} + \mathbf{w}
\end{equation}
where $\mathbf{x} \in \mathbb{R}^{(N+1)n_x}$ and $\mathbf{u} \in
\mathbb{R}^{n}$, $n=Nn_u$, stack state and control variables over the horizon.
The variable $\mathbf{w}$ collects the accumulated contributions of the
affine dynamics terms $w_t$. Substituting the expression for $\mathbf{x}$ in
the constraints and cost yields the condensed QP
\begin{equation}
\label{prob:condensed_prob}
\begin{array}{ll}
        \underset{\mathbf{u}}{\mbox{minimize}} & \frac{1}{2} \mathbf{u}^{\top} H \mathbf{u} + g^{\top} \mathbf{u} \\
        \mbox{subject to} & G \mathbf{u} \leq b \\
                          & E \mathbf{u} = e \,.
\end{array}
\end{equation}
Under the cost assumptions for Problem~\eqref{prob:ocp_problem}, $H \succ 0$,
and after the equality preprocessing described above $E$ has full row rank.
Introducing dual variables only for the inequality constraints, the Lagrangian
of Problem~\eqref{prob:condensed_prob} is
\begin{equation}
\label{eq:lagrangian}
\mathcal{L}(\mathbf{u}, \lambda) = \frac{1}{2} \mathbf{u}^{\top} H \mathbf{u} + g^{\top} \mathbf{u} + \mathcal{I}_0(E\mathbf{u}-e) + \lambda^{\top}(G\mathbf{u}-b)
\end{equation}
where $\mathcal{I}_0(z)$ is zero if $z=0$ and $+\infty$ otherwise. If
Problem~\eqref{prob:condensed_prob} is feasible, strong duality for convex QPs
with affine constraints allows us to recover its solution from the dual
\begin{equation}
\label{prob:dual_prob}
\begin{array}{ll}
        \underset{\lambda \geq 0}{\mbox{maximize}} & \min_{\mathbf{u}} \mathcal{L}(\mathbf{u}, \lambda).
\end{array}
\end{equation}
From the equality-constrained QP optimality conditions and the block
inverse formula~\cite[Sec.~10.1.1 and App.~A.5.5]{BV2004}, we find that
Problem~\eqref{prob:dual_prob} is equivalent to the convex QP
\begin{equation}
\label{prob:dual_prob_qp}
\begin{array}{ll}
        \underset{\lambda \geq 0}{\mbox{minimize}} & \displaystyle \frac{1}{2} \lambda^{\top} \mathcal{H} \lambda + f^{\top} \lambda \,.
\end{array}
\end{equation}
where
\begin{align}
\label{eq:dual_hessian}
\mathcal{H} &= G\left(H^{-1} - H^{-1}E^{\top}(EH^{-1}E^{\top})^{-1} EH^{-1}\right)G^{\top} \\
f &= b - G \mathbf{u}_0 \tag*{}
\end{align}
with $\mathbf{u}_0 = \arg \min_{\mathbf{u}} \mathcal{L}(\mathbf{u}, 0)$. Note
that the matrix in parentheses in~\eqref{eq:dual_hessian} is the upper-left
block of the inverse equality-constrained KKT matrix. 

Algorithm~\ref{algo:AS_algorithm} describes the AS algorithm employed by DAOCP.
A working set $\mathcal{W} \subseteq \{1,\dots,m\}$ represents an equality
constrained QP obtained from~\eqref{prob:dual_prob_qp}, where all $\lambda_i$
with $i\in \{1,\dots,m\}\setminus \mathcal{W}$ are constrained through
$\lambda_i = 0$, while the remaining components of $\lambda$ are free. Note
that we are using the convention of representing the working set through
constraints that are active in the primal problem, as in~\cite{G1983}. We
denote by $\mathcal{H}_{\mathcal{W}}$ the matrix obtained from the rows and
columns of $\mathcal{H}$ with index in $\mathcal{W}$. Similarly,
$\lambda_{\mathcal{W}}$ denotes the vector of components of $\lambda$
indexed by $\mathcal{W}$. 

\begin{algorithm}[t]
\caption{The Dual Active Set algorithm. DAOCP exploits OCP structure in the operations colored in blue.}
\label{algo:AS_algorithm}
\footnotesize
\begin{algorithmic}[1]
\State $\mathcal{W}$, $\lambda$ given.
\Loop
    \If{$\mathcal{H}_{\mathcal{W}}$ is nonsingular}
        \State Solve \textcolor{blue}{$\mathcal{H}_{\mathcal{W}}\nu=-f_{\mathcal{W}}$}
        \If{$\nu \geq 0$}
            \State $\lambda_{\mathcal{W}}\gets\nu$;
                \textcolor{blue}{$\mathbf{u}\gets\arg\min_{v}
                \mathcal{L}(v,\lambda)$}
            \State \textbf{if} \textcolor{blue}{$G\mathbf{u}\leq b$}
                \textbf{ then \Return} $\mathbf{u}$
            \State $i\gets\arg\max_{j\notin\mathcal{W}}
                (G\mathbf{u}-b)_j$
            \State $\mathcal{W}\gets\mathcal{W}\cup\{i\}$
        \Else
            \State $d\gets\nu-\lambda_{\mathcal{W}}$
            \State $i\gets\arg\min_{j\in\mathcal{W}:\,d_j<0}
                \left(-\lambda_j/d_j\right)$
            \State $t\gets-\lambda_i/d_i$;
                $\lambda_{\mathcal{W}}\gets\lambda_{\mathcal{W}}+td$;
                $\mathcal{W}\gets\mathcal{W}\setminus\{i\}$
        \EndIf
    \Else
        \State Find $p$ such that \textcolor{blue}{$\mathcal{H}_{\mathcal{W}}p=0$}
            and $p^{\top}f_{\mathcal{W}}<0$
        \State \textbf{if} $p\geq0$
            \textbf{ then \Return} \textsc{infeasible}
        \State $i\gets\arg\min_{j\in\mathcal{W}:\,p_j<0}
            \left(-\lambda_j/p_j\right)$
        \State $t\gets-\lambda_i/p_i$;
            $\lambda_{\mathcal{W}}\gets\lambda_{\mathcal{W}}+tp$;
            $\mathcal{W}\gets\mathcal{W}\setminus\{i\}$
    \EndIf
\EndLoop
\end{algorithmic}
\end{algorithm}
Algorithm~\ref{algo:AS_algorithm} is essentially identical to the algorithm
employed by DAQP, with the equality constraints eliminated here before applying
the active-set procedure. We refer to~\cite{DA2022,G1983} for the correctness
and convergence analysis, including the treatment of degeneracy.

\subsection{Key implementation aspects}
In this subsection we detail the key operations in
Algorithm~\ref{algo:AS_algorithm} which remain implicit in its mathematical
formulation.

\textit{1) Solving unconstrained QPs:} in order to solve the linear systems for
$\nu$ and $p$ in Algorithm~\ref{algo:AS_algorithm}, we maintain a Cholesky
factorization of $\mathcal{H}_{\mathcal{W}}$. When a constraint is added, we
extend the factor by one row, while when one is removed we delete the
corresponding row and column and perform a rank-one update. DAQP uses analogous
updates to an $LDL^{\top}$ factorization.

Note that in order to maintain the Cholesky factor of $\mathcal{H}_{\mathcal{W}}$ we do not need knowledge of the full $\mathcal{H}$. Instead we only need to be able to compute the entries of $\mathcal{H}$ that constitute a new row of $\mathcal{H}_{\mathcal{W}}$, after a constraint is added to $\mathcal{W}$.

To this end, DAQP computes a Gramian
factorization of $\mathcal{H}$ during solver setup. For the
inequality-only formulation in~\cite{DA2022}, $\mathcal{H} = GH^{-1}G^{\top}$,
and the chosen factorization is $\mathcal{H} = MM^{\top}$ with $M = G
L_H^{-\top}$, where $L_H = \operatorname*{chol}(H)$. A Gramian factorization is
useful because it allows us to compute new rows of
$\mathcal{H}_{\mathcal{W}}$ by simply evaluating a matrix-vector product. 

DAOCP instead employs a \textit{signed} Gramian factorization of the form
\begin{equation}
\label{eq:signed_gramian} 
\mathcal{H} = M_u M_u^{\top} - M_{\xi} M_{\xi}^{\top} \,.
\end{equation}
Moreover, we do not compute the matrices $M_u$ and
$M_{\xi}$ in full during solver setup. This is due to the fact that we
target nonlinear or time-varying MPC applications, where changes in the problem
matrices prevent reuse across sampling times. We therefore instead
compute only the rows of $M_u$ and $M_{\xi}$ that are actually needed, while the algorithm is running.

The following Proposition characterizes the signed Gramian factorization that we employ.

\begin{proposition}
\label{prop:signed_gramian}
Let 
\begin{equation}
\label{eq:KKT_matrix}
\mathcal{K} = \begin{bmatrix}
        H & E^{\top} \\
        E & 0
\end{bmatrix}
\end{equation}
be the KKT matrix associated to the equality constraints of Problem~\eqref{prob:condensed_prob}, and let $L$ be any matrix such that
\[
\mathcal{K} = L \operatorname*{blockdiag}(I_n, - I_{\rho}) L^{\top} \,.
\]
Then the dual Hessian $\mathcal{H}$ admits the signed Gramian factorization
\[
\mathcal{H} = M_u M_u^{\top} - M_{\xi} M_{\xi}^{\top} \,,
\]
where $M_u$ and $M_{\xi}$ solve the linear system
\begin{equation}
\label{eq:MuMxi_linear_system}
L \begin{bmatrix}
        M_u^{\top} \\
        M_{\xi}^{\top}
\end{bmatrix} = \begin{bmatrix}
        G^{\top} \\
        0
\end{bmatrix} \,.
\end{equation}
\end{proposition}
\begin{proof}
The statement follows from the observation that $\mathcal{H}$ is the top-left
block of $\mathcal{K}^{-1}$, multiplied by $G$ on the left and $G^{\top}$ on
the right.
\end{proof}
The left factor $L$ of the KKT matrix $\mathcal{K}$ defining our Gramian
factorization may be chosen in different ways. For example, one option is to
compute $L$ through an LDL factorization of $\mathcal{K}$. Later we describe
how $L$ can be selected so that linear systems in the form of
Equation~\eqref{eq:MuMxi_linear_system} can be solved efficiently.

\textit{2) Computing primal candidates:} another key operation performed in
Algorithm~\ref{algo:AS_algorithm} is the retrieval of the primal candidate
associated to the current dual multiplier $\lambda$, which we need in order to
compute the residuals $b - G\mathbf{u}$, to determine whether a constraint
should be added to the working set. While one could retrieve $\mathbf{u}$ by
directly minimizing the Lagrangian, we choose to do that through the
relationship established by the following proposition.
\begin{proposition}
Let $\mathcal{K}$, $L$, $M_u$, $M_{\xi}$ be defined as in
Proposition~\ref{prop:signed_gramian}. For a pair $\left(\mathcal{W},
\lambda\right)$ generated by Algorithm~\ref{algo:AS_algorithm}, the minimizer
is $\mathbf{u}=\mathbf{u}_0+\Delta\mathbf{u}$, with $\Delta \mathbf{u}$ solving
the linear system
\begin{equation}
\label{eq:primal_candidate}
L^{\top} \begin{bmatrix}\Delta\mathbf{u} \\ \Delta\boldsymbol{\xi}\end{bmatrix} = \begin{bmatrix}
-M_{u,\mathcal{W}}^{\top} \\
M_{\xi, \mathcal{W}}^{\top}
\end{bmatrix} \lambda_{\mathcal{W}} \,,
\end{equation}
for some $\Delta\boldsymbol{\xi}$.
\end{proposition}
\begin{proof}
The statement follows from the definitions of $\mathcal{L}$, $L$, $M_{u}$ and $M_{\xi}$.
\end{proof}
\textit{3) Constraint selection:} once the primal candidate is computed, we
scan the residual $b - G\mathbf{u}$ to determine which constraint should be
added to the working set (if any). Classically, one may employ one of two
possible heuristics: either all residuals are evaluated and the constraint with
maximum violation is added, or instead we can adopt a \textit{greedy} approach,
where we stop as soon as we find a violated constraint. The most-violated rule
can reduce the number of iterations, while the greedy rule can reduce the cost
of constraint selection. Their relative performance depends on the problem.

For simplicity, we present results for the
\textit{most-violated} heuristic, which is the one chosen by DAQP, but our implementation allows the user to
select which strategy to adopt.

\subsection{Exploiting OCP structure}
The algorithm described so far applies to a generic QP of the
form~\eqref{prob:condensed_prob}. Now instead we describe how to rely on OCP structure to implement its main operations efficiently.

The key observation is that solving linear systems with the KKT matrix
$\mathcal{K}$ in~\eqref{eq:KKT_matrix} is equivalent to solving finite-horizon
LQR problems with stage-wise equality constraints. In particular, for the
special case of a condensed QP with no equality constraints, where $L$ is the
Cholesky factor of $H$, \cite{GF2013_TRCHOL} established that solving $L x = b$
is equivalent to running the backward co-state LQR recursion~\cite[Algorithm 2]{GF2013}, while solving $L^{\top}y = x$ is equivalent
to running the forward LQR recursion~\cite[Algorithm 2]{GF2013}. Here
we generalize this result to a signed Cholesky, or LDL, factorization of
$\mathcal{K}$.

Before introducing our main results, we define a specific form of the Riccati
recursion, which is associated to OCPs in the form of
Problem~\eqref{prob:ocp_problem_eqcon}. This recursion, which is tightly related
to the one in~\cite{LV2024}, follows the classical dynamic programming argument,
with the cost-to-go minimization subject to
$D_t^x x_t + D_t^u u_t = d_t$. These equality constraints can be satisfied for
any value of $x_t$, since $D_t^u$ is full row rank by assumption. In
particular, we introduce the quantities
\begin{equation}
\label{eq:eqcon_lqr_Lambda}
\Lambda_t = \operatorname*{signedchol} \left(\begin{bmatrix}
        R_t + B_t^{\top} P_{t+1} B_t & (D^u_t)^{\top} \\
        D_t^u & 0
\end{bmatrix}\right)
\end{equation}
\begin{equation}
\label{eq:eqcon_lqr_K}        
\begin{bmatrix}K_t^u & K_t^{\xi}\end{bmatrix}
= \begin{bmatrix}S_t^{\top} + A_t^{\top} P_{t+1} B_t & (D_t^x)^{\top}\end{bmatrix}
\Lambda_t^{-\top}
\end{equation}
\begin{equation}
\label{eq:eqcon_lqr_riccati}
P_t = Q_t + A_t^{\top}P_{t+1}A_t - K^u_t(K_t^u)^{\top} + K_t^{\xi}(K_t^{\xi})^{\top}
\end{equation}
where $P_N = Q_N$. The $\operatorname*{signedchol}(\cdot)$ operator is
defined to take as input the KKT matrix of a
strictly convex, equality constrained QP
\[
\begin{aligned}
\operatorname*{minimize}_z \, \frac{1}{2} z^{\top} \mathcal{Q} z + q^{\top} z \quad
\text{s.t.}\quad \mathcal{E} z = b
\end{aligned}
\]
with full row rank $\mathcal{E}$, and return a matrix
factor $\Lambda$ such that 
\begin{equation}
\label{eq:singed_cholesky}
\Lambda = \begin{bmatrix}
        \operatorname*{chol}(\mathcal{Q}) & 0 \\
        \mathcal{E}\operatorname*{chol}(\mathcal{Q})^{-\top} & \operatorname*{chol}(\mathcal{E} \mathcal{Q}^{-1} \mathcal{E}^{\top})
\end{bmatrix} \,.
\end{equation}

Recall that the main use of Equation~\eqref{eq:signed_gramian} is to compute
the row of $\mathcal{H}_{\mathcal{W}}$ associated to a newly added constraint. Assume
that this is inequality constraint $j$ at time step $t$. This in turn requires
computing the corresponding rows of $M_u$ and $M_{\xi}$, which we denote by
$m_u$ and $m_{\xi}$, respectively. The following theorem states
that by suitably choosing $L$ in Proposition~\ref{prop:signed_gramian},
$m_u$ and $m_{\xi}$ can be computed by running a generalization of the
classical backward LQR recursion.

\begin{theorem}
\label{thm:MuMxi_rows_computation}
Let $(C_t^u)_j$ and $(C_t^x)_j$ be the rows of $C_t^u$ and $C_t^x$
corresponding to the newly added constraint. Partition $m_u$ and $m_{\xi}$ as
\[
m_u = \begin{bmatrix}
        y_{0}^{\top} & y_1^{\top} & \dots & y_{N-1}^{\top}
\end{bmatrix}^{\top}
\]
\[
m_{\xi} = \begin{bmatrix}
        v_{0}^{\top} & v_1^{\top} & \dots & v_{N-1}^{\top}
\end{bmatrix}^{\top}
\]
where $y_t \in \mathbb{R}^{n_u}$ and $v_t \in \mathbb{R}^{\rho_t}$. There
exists a factorization of the KKT matrix~\eqref{eq:KKT_matrix} of the form
\[
\mathcal{K} = L \operatorname*{blockdiag}(I_n, -I_{\rho}) L^{\top}
\]
such that the linear system~\eqref{eq:MuMxi_linear_system} defining $m_u$ and
$m_{\xi}$ can be solved through the following recursion
\begin{equation}
\label{eq:dual_recursion}
\begin{aligned}
\begin{bmatrix}y_t \\ v_t\end{bmatrix}
&= \Lambda_t^{-1}\begin{bmatrix}(C_t^u)_j \\ 0\end{bmatrix},\\
p_t &= (C_t^x)_j - K_t^u y_t + K_t^{\xi} v_t .
\end{aligned}
\end{equation}
\begin{equation}
\label{eq:thm1_backward_rec}
\begin{aligned}
\begin{bmatrix}y_{\tau} \\ v_{\tau}\end{bmatrix}
&= \Lambda_{\tau}^{-1}\begin{bmatrix}B_{\tau}^{\top}p_{\tau+1} \\ 0\end{bmatrix},\\
p_{\tau} &= A^{\top}_{\tau}p_{\tau+1} - K_{\tau}^u y_{\tau} + K_{\tau}^{\xi}v_{\tau}
\end{aligned}
\end{equation}
for $\tau=t-1,\dots,0$, while $v_{\tau}=0$ and $y_{\tau}=0$ for $\tau>t$. For a
terminal constraint ($t=N$), initialize $p_N=(C_N^x)_j$ and apply~\eqref{eq:thm1_backward_rec} for $\tau=N-1,\dots,0$.
\end{theorem}
\begin{proof}
The theorem is proved in Appendix~\ref{app:proof}. The proof generalizes the
argument in~\cite{GF2013_TRCHOL} to equality constrained QPs.
\end{proof}
A similar result applies to the problem of solving the linear system that
defines primal candidates in Equation~\eqref{eq:primal_candidate}. In this
case, the system can be solved by running a generalization of the forward LQR
recursion.
\begin{theorem}
\label{thm:primal_candidate}
For an arbitrary working set $\mathcal{W}$ and dual multiplier $\lambda$, partition
\[
M_{u,\mathcal{W}}^{\top} \lambda_{\mathcal{W}} = \begin{bmatrix}
        \delta u_{0}^{\top} & \delta u_1^{\top} & \dots & \delta u_{N-1}^{\top}
\end{bmatrix}^{\top}
\]
\[
M_{\xi,\mathcal{W}}^{\top} \lambda_{\mathcal{W}} = \begin{bmatrix}
        \delta \xi_{0}^{\top} & \delta \xi_1^{\top} & \dots & \delta \xi_{N-1}^{\top}
\end{bmatrix}^{\top} \,.
\]
Then the primal candidate $\mathbf{u} = \arg \min_v \mathcal{L}(v, \lambda)$
can be written as $\mathbf{u} = \mathbf{u}_0 + \Delta \mathbf{u}$, where the
time-indexed blocks of $\Delta \mathbf{u}$ are obtained from the following
recursion
\begin{equation}
\label{eq:primal_recusion}
\begin{aligned}
\begin{bmatrix}
        \Delta u_t \\ \Delta \xi_t
\end{bmatrix} &= \Lambda_t^{-\top} \begin{bmatrix}
        -(K_t^u)^{\top} \Delta x_t - \delta u_t \\
        (K_t^{\xi})^{\top} \Delta x_t + \delta \xi_t
\end{bmatrix} \\
\Delta x_{t+1} &= A_t \Delta x_t + B_t \Delta u_t 
\end{aligned}
\end{equation}
with $\Delta x_0 = 0$.
\end{theorem}
\begin{proof}
The statement follows from a similar argument as in the proof of Theorem~\ref{thm:MuMxi_rows_computation}.
\end{proof}
\begin{remark}
Notice that when computing a primal candidate satisfying
Equation~\eqref{eq:primal_candidate} through the recursion stated by
Theorem~\ref{thm:primal_candidate}, we must also compute the corresponding
state sequence. This allows us to evaluate the constraint residuals $b -
G\mathbf{u}$ without having to form $b$ or $G$.
\end{remark}

Theorems~\ref{thm:MuMxi_rows_computation} and~\ref{thm:primal_candidate}
together imply that we do not need to form the condensed
Problem~\eqref{prob:condensed_prob}, since all operations in
Algorithm~\ref{algo:AS_algorithm} can be expressed in terms of the original OCP
and the associated Riccati recursion.

\subsection{Analyzing complexity}
The most costly operations performed by the algorithm are
recursions~\eqref{eq:dual_recursion}--\eqref{eq:primal_recusion}, as well as
the linear system solves involving $\mathcal{H}_{\mathcal{W}}$ as coefficient
matrix, and the matrix vector products $M_{u,\mathcal{W}}^{\top} \lambda_{\mathcal{W}}$,
$M_{\xi, \mathcal{W}}^{\top} \lambda_{\mathcal{W}}$. Overall the iteration cost is
\begin{equation}
\label{eq:c_iter}
\mathcal{C}_{\text{iter}} = \mathcal{C}_{\text{AS}} + \mathcal{C}_{\text{LQR}} + \mathcal{C}_{\text{scan}}\,,
\end{equation}
where
\begin{equation}
\label{eq:c_lqr}
\begin{aligned}
\mathcal{C}_{\text{AS}} &= \mathcal{O}\left(|\mathcal{W}|^2 + |\mathcal{W}|(\rho + Nn_u)\right) \\
\mathcal{C}_{\text{LQR}} &= \mathcal{O}\left(N(n_x^2+n_u^2) + \sum_{t=0}^{N-1}\rho_t^2\right) \\
\mathcal{C}_{\text{scan}} &= \mathcal{O}\left(m(n_x + n_u)\right) \,.
\end{aligned}
\end{equation}
The quadratic scaling of $\mathcal{C}_{\text{AS}}$ with $|\mathcal{W}|$ means
that iterations may become quite expensive when the working set is large. This
is the crucial reason to exclude equality constraints from the AS algorithm: if
we simply keep all equalities in the working set, then at every
iteration $\mathcal{C}_{\text{AS}} \geq \rho^2$. Instead, with our approach, the
quadratic cost of equality constraints is spread through the horizon as per
Equation~\eqref{eq:c_lqr} (notice that $\rho^2 \geq \sum_{t=0}^{N-1}\rho_t^2$).

Nevertheless, the quadratic scaling of $\mathcal{C}_{\text{AS}}$ with $|\mathcal{W}|$
implies that DAOCP is optimized for problems where the working set does not
grow too large. When this design assumption fails to hold, it becomes more
convenient to formulate dual subproblems as equality constrained LQR problems,
implying that the cost term which is quadratic in $|\mathcal{W}|$
becomes $\mathcal{O}(\sum_{t=0}^N |\mathcal{W}_t|^2)$, where $\mathcal{W}_t$
denotes the set of active primal constraints at time $t$.

We find the approach taken by DAOCP to be very competitive on MPC problems, where,
after eliminating equality constraints, our assumption on the size of the
working set often holds. In future work we will explore the formulation briefly
described above.

\subsection{Solver implementation}

DAOCP is implemented in \texttt{C}, and relies on BLASFEO for a
significant portion of the numerical linear algebra. BLASFEO was chosen because
it is specifically designed to deliver high performance on Riccati and LQR-like
routines, where several operations are performed on relatively small matrices.
The solver code is available at https://github.com/AlbertoZaupa/daocp.

\section{Numerical Experiments}
\label{sec:numerical_experiments}
In this section we compare DAOCP against other state of the art solvers for MPC
on three robotics benchmarks. The first involves stabilizing the full-body model
of an Atlas humanoid robot, taken from~\cite{AB2024}. Then we reproduce the
benchmark proposed in~\cite{FS2025}, where a nonlinear MPC controller follows
time-varying reference trajectories for a quadruped robot, both in standing and
trotting scenarios. Finally, we compare the solvers on two challenging
reference tracking problems for quadrotors from~\cite{SS2022}. 

The Atlas benchmark is a linear time-invariant MPC problem, where the robot
dynamics are linearized about a reference equilibrium. The quadruped and
quadrotor benchmarks instead involve a nonlinear MPC controller in Real Time
Iteration (RTI)~\cite{SG2016} fashion, implemented in Acados.

We compare against DAQP~\cite{DA2022}, HPIPM~\cite{GF2020} and
qpOASES~\cite{JF2014}. 
The Interior Point method HPIPM is the most established OCP solver,
while DAQP and
qpOASES are the most popular and competitive AS solvers for MPC. For all
benchmarks the termination tolerance of HPIPM was set to $10^{-3}$, and the
chosen mode was \texttt{speed\_abs}. This configuration allowed HPIPM to
deliver the best possible solve times, with no observable degradation in
control performance. Otherwise all solvers were given default parameters. All
benchmarks were run on a laptop equipped with an Intel Core Ultra 7 265U processor.

\subsection{Atlas stabilization}
The dynamics of the Atlas robot are described by a state space model with 58 state
variables and 29 controls~\cite{AB2024}. The horizon length for this problem is $N=30$. Because the controller uses the linear model corresponding to the reference equilibrium, the reported solve times do not include factorization costs, since they can be reused across sampling times. 

The closed-loop simulation is run
50 times, and the per sample timings are averaged over all runs. We report in
Table~\ref{tab:atlas_timings} the average and worst case solve times over
simulation steps. Excluding qpOASES, which does not perform particularly well on
this benchmark, the results are in agreement with the typical differences
between Interior Point and Active Set methods: while DAQP and DAOCP are much
faster than HPIPM on average, the difference is less significant in the worst
case. Nevertheless, DAOCP outperforms all other solvers on both metrics. 

The gap in solve times between DAOCP and DAQP, which essentially implement the same high-level AS algorithm, highlights how being able to
exploit OCP structure is critical for this particular benchmark. In fact the horizon length $N=30$ greatly exceeds the thresholds $N_s(2) = 13$, $N_f(2) \approx 10.8$ established in Section~\ref{sec:when_ocp_matters}, and in particular the flop count of LQR solves is less than half of dense triangular solves. This is a key contribution to the 9$\times$ improvement in average solve time for DAOCP, which is further explained by the difference in linear algebra backend. 
\begin{table}[t]
    \caption{Atlas benchmark solve times.}
    \label{tab:atlas_timings}
    \centering
    \footnotesize
    \begin{tabular}{lrrrr}
        \toprule
        & qpOASES & HPIPM & DAQP & DAOCP \\
        \midrule
        Average [ms] & 3.64 & 1.75 & 0.27 & \textbf{0.03} \\
        Worst case [ms] & 176.59 & 4.99 & 3.82 & \textbf{1.08} \\
        \bottomrule
    \end{tabular}
\end{table}

\subsection{Quadruped control}
\begin{figure}[!t]
    \centering
    \includegraphics[width=\columnwidth]{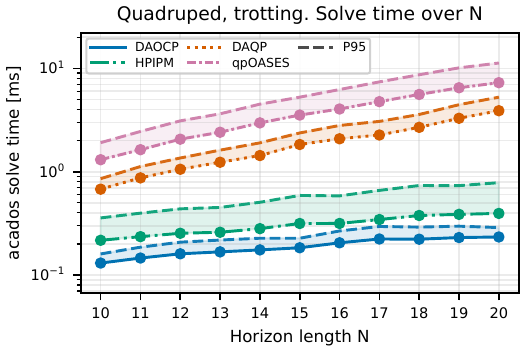}
    \caption{Average and p95 solve times for the quadruped benchmark in the trotting scenario. $N$ ranges from $10$ to $20$. The trends are consistent with the $N$-linear scaling of HPIPM and DAOCP, and the hybrid quadratic-cubic scaling of DAQP and qpOASES.}
    \label{fig:quadruped_horizon_sweep}
\end{figure}

In the quadruped control benchmark from~\cite{FS2025}, the robot is modeled as a
rigid body, and the control inputs are the ground-contact forces exerted
through the feet. The system state has 13 components, while the control input
has 12. Because $n_u$ and $n_x$ are relatively close, exploiting OCP structure
is again particularly favorable, and in fact the horizon thresholds from Section~\ref{sec:when_ocp_matters} are quite low at $N_s,N_f \approx 5.5$.

The benchmark involves two types of motion: standing, where the controller is
given time-varying body orientation references, and trotting, where the
quadruped is walking, following a time-varying velocity reference. In the
trotting scenario, the ability to deal efficiently with stage-wise equality
constraints is crucial, since at every sampling time the feet that are not in
contact with the ground are constrained to exert null force. All solvers have this capability except qpOASES, and this in part explains its gap in performance compared with DAQP.

Runtime measurements are averaged on a per-sampling time basis over 5
independent runs. Figure~\ref{fig:quadruped_horizon_sweep} shows the average and
p95 solve times for the four solvers in the trotting scenario, for horizon
lengths ranging from $N=10$ to $N=20$. The results highlight how HPIPM and DAOCP
scale more favorably with horizon length: over the given range of $N$, their
average solve time roughly doubles, while those of DAQP and qpOASES increase by more than a
factor of $4$. This is consistent with the inherent linear scaling in $N$ of
structure exploiting methods, and the quadratic scaling of dense AS iterations. Table~\ref{tab:quadruped_tab} reports average
and worst case solve times in both scenarios for $N=20$.

\begin{table}[t]
    \caption{Quadruped solve times, $N=20$.}
    \label{tab:quadruped_tab}
    \centering
    \footnotesize
    \begin{tabular}{llrr}
        \toprule
        Solver & Motion & Average [ms] & Worst case [ms] \\
        \midrule
        \multirow{2}{*}{HPIPM}
          & Standing & 0.240 & 0.443 \\
          & Trotting & 0.391 & 0.804 \\
        \addlinespace[1pt]
        \multirow{2}{*}{DAQP}
          & Standing & 2.575 & 3.835 \\
          & Trotting & 3.544 & 5.579 \\
        \addlinespace[1pt]
        \multirow{2}{*}{qpOASES}
          & Standing & 3.815 & 19.440 \\
          & Trotting & 7.332 & 11.579 \\
        \addlinespace[1pt]
        \multirow{2}{*}{DAOCP}
          & Standing & \textbf{0.188} & \textbf{0.329} \\
          & Trotting & \textbf{0.235} & \textbf{0.427} \\
        \bottomrule
    \end{tabular}
\end{table}

\subsection{Quadrotor reference tracking}
\begin{table}[t]
    \caption{Quadrotor solve times. Dynamically feasible trajectory.}
    \label{tab:quadrotor_feasible_timings}
    \centering
    \footnotesize
    \begin{tabular}{lrrrr}
        \toprule
        & qpOASES & HPIPM & DAQP & DAOCP \\
        \midrule
        Average [$\upmu$s] & 132.12 & 87.04 & 111.79 & \textbf{81.80} \\
        Worst case [$\upmu$s] & 143.90 & 100.53 & 117.70 & \textbf{88.03} \\
        \bottomrule
    \end{tabular}
\end{table}

\begin{table}[t]
    \caption{Quadrotor solve times. Dynamically infeasible trajectory.}
    \label{tab:quadrotor_infeasible_timings}
    \centering
    \footnotesize
    \begin{tabular}{lrrrr}
        \toprule
        & qpOASES & HPIPM & DAQP & DAOCP \\
        \midrule
        Average [$\upmu$s] & 354.81 & \textbf{150.62} & 181.32 & 282.01 \\
        Worst case [$\upmu$s] & 1060.89 & \textbf{227.94} & 331.81 & 540.51 \\
        \bottomrule
    \end{tabular}
\end{table}

In this benchmark we aim at testing the performance
of DAOCP on a problem formulation where exploiting OCP structure is not clearly advantageous. The problem dimensions are $N=20$, $n_x=13$ and $n_u = 4$. From Section~\ref{sec:when_ocp_matters}, we get $N_f \approx 20$ and $N_s \approx 29$, and considering the relatively small value of $n_u$, condensing based methods are expected to have an advantage, especially when it comes to linear system solves and similar operations relevant for AS solvers.

In~\cite[Eq.~36]{SS2022} the authors define a family of 144 reference trajectories
and we compare the four solvers on two instances of different levels of difficulty. The first is obtained by setting $V_{\text{max}} = 15$ m/s, $a_{\text{max}} = 30$ m/s$^2$, $n=1$, while the second with $V_{\text{max}} = 15$ m/s,
$a_{\max} = 60$ m/s$^2$ and $n=2$. The first trajectory was selected as one of the most challenging from the dynamically feasible subset, while the other, which is dynamically infeasible, was chosen specifically as the one where DAOCP performed the worst relative to the other solvers.

Tables~\ref{tab:quadrotor_feasible_timings} and~\ref{tab:quadrotor_infeasible_timings} report the runtime breakdown for the two trajectories.  
Timing results were averaged on a per simulation step
basis, across 400 individual runs. DAOCP is the fastest solver on the feasible trajectory, while it is outperformed by DAQP and HPIPM on the infeasible one. Across all 144 instances from~\cite{SS2022}, DAOCP was the fastest in terms of average and worst-case solve times in $40$\% and $49$\% of the cases, respectively. Instead HPIPM was the fastest on average $54$\% of the times, and $40$\% in the worst case. In the remaining few cases DAQP was the fastest, while qpOASES was never faster than all three other solvers. 

While dense solvers were expected to have an edge given the problem dimensions, dedicated OCP solvers were generally faster. The reason for this discrepancy is that the amount of constant work that structure exploiting solvers perform at every sample time is less than for dense solvers. This point becomes clear when comparing the performance of DAOCP and DAQP on the two trajectories. At every sample time DAQP needs to condense the problem and compute the full Gramian factorization $\mathcal{H}=MM^{\top}$, while DAOCP can start solving the problem right away. However, as established above, AS iterations are more expensive for DAOCP, which is therefore slower when several iterations are needed before reaching convergence, as is the case on the infeasible trajectory.

\section{Conclusion}
\label{sec:conclusion}
In this paper we presented DAOCP, a dual Active Set solver for linear quadratic
OCPs. The solver is inspired by the established QP solver DAQP, and it improves
its performance on OCP-structured problems by relying on the relationship between the Riccati recursion of an LQR problem with
stage-wise equality constraints and a specific LDL factorization of the
corresponding condensed KKT matrix.

DAOCP is compared with HPIPM, DAQP and qpOASES on a set of robotics-inspired
benchmarks, where it showcases compelling performance, combining the desirable
features of Active Set methods and the advantages of OCP solvers in dimensionality regimes typical of robotics.

\onecolumn
\appendix
\subsection{Proof of Theorem~\ref{thm:MuMxi_rows_computation}}
\label{app:proof}
Before proving the statement of Theorem~\ref{thm:MuMxi_rows_computation} in full generality, we consider the special case $N=2$. This allows us to more clearly illustrate the key mechanisms that are employed in the general proof that follows.

\begin{proof}[Proof of Theorem~\ref{thm:MuMxi_rows_computation}, N=2]
We begin by writing the KKT matrix $\mathcal{K}$ of Equation~\eqref{eq:KKT_matrix} after introducing the following time-based block ordering for its rows and columns
\[
u_{N-1},\, \xi_{N-1}, \, u_{N-2}, \, \xi_{N-2}, \, \dots, \, u_0, \, \xi_0
\]
where $\xi_t$ blocks are associated to the rows of $E$ that correspond to constraints at time $t$. We denote by $\bar{\mathcal{K}}$ the KKT matrix in this basis, which reads
\begin{gather}
\label{eq:KKT_matrix_time_rev}
        \bar{\mathcal{K}} = \begin{bmatrix}
                R_1 + B_1^{\top}Q_2 B_1 & (D_1^u)^{\top} & & \\
                D_1^u & 0 & & \\
                B_0^{\top} S_1^{\top} + B_0^{\top} A_1^{\top} Q_2 B_1 & B_0^{\top} (D^x_1)^{\top} & R_0 + B_0^{\top} Q_1 B_0 + B_0^{\top}A_1^{\top} Q_2 A_1 B_0 & (D_0^u)^{\top} \\ %u_0
                0 & 0 & D_0^u & 0
        \end{bmatrix}
\end{gather}
where the top-right block is left blank because the matrix is symmetric. In the
following, the submatrix of $\bar{\mathcal{K}}$ obtained from the rows
associated to $(u_{\tau}, \xi_{\tau})$ and the columns associated to $(u_t,
\xi_t)$ will be denoted as the $(\tau, t)$-block of
$\bar{\mathcal{K}}$.

We are now going to compute an LDL factorization of
$\bar{\mathcal{K}}$ of the form
\[
\bar{\mathcal{K}} = \bar{L} \Sigma \bar{L}^{\top}
\]
where $\Sigma$ is diagonal and $\bar{L}$ is lower triangular and invertible.
This factorization is guaranteed to exist because $H \succ 0$ and $E$ is full
row-rank. In order to compute such an LDL factorization, we employ the
classical factorization procedure which starts by identifying a block
partitioning of the form 
\[
\bar{\mathcal{K}} = \begin{bmatrix}
        \bar{\mathcal{K}}_{1,1} & \\
        \bar{\mathcal{K}}_{0,1} & \bar{\mathcal{K}}_{0,0}
\end{bmatrix}
\]
chosen so that we know how to factorize $\bar{\mathcal{K}}_{1,1} =
\bar{L}_{1,1} \Sigma_1 \bar{L}_{1,1}^{\top}$. We then begin to fill in the LDL
of $\bar{\mathcal{K}}$ as follows
\[
\bar{L} = \begin{bmatrix}
        \bar{L}_{1,1} &  \\
        \bar{L}_{0,1} & \star
\end{bmatrix}, \quad \Sigma = \operatorname*{blockdiag}(\Sigma_1, \star)
\]
where $\bar{L}_{0,1} = \bar{\mathcal{K}}_{0,1} \bar{L}_{1,1}^{-\top}
\Sigma_1^{-1}$, and the blocks denoted by $\star$ will be computed in the
following steps of the procedure. We then update the bottom-right block of
$\bar{\mathcal{K}}$ as $$\bar{\mathcal{K}}_{0,0} \leftarrow
\bar{\mathcal{K}}_{0,0} - \bar{L}_{0,1}\Sigma_1 \bar{L}_{0,1}^{\top}$$ and
iterate the same steps. We refer to the $\bar{\mathcal{K}}_{0,0}$ update as
the \textit{Schur complement update}, because $\bar{\mathcal{K}}_{0,0} -
\bar{L}_{0,1} \Sigma_1 \bar{L}_{0,1}^{\top}$ is the Schur complement of
$\bar{\mathcal{K}}$ with respect to $\bar{\mathcal{K}}_{1,1}$.

We now apply the described procedure to the matrix $\bar{\mathcal{K}}$ in
Equation~\eqref{eq:KKT_matrix_time_rev}. Notice that the $(1,1)$-block of
$\bar{\mathcal{K}}$ is a KKT matrix with standard variable ordering, with
positive definite Hessian and full row-rank constraint matrix. It therefore
admits a signed Cholesky factorization (as defined in
Equation~\eqref{eq:singed_cholesky}). Define $P_2 = Q_2$ and 
\[
\begin{aligned}
\Lambda_1 &= \operatorname*{signedchol}\left(\begin{bmatrix}
        R_1 + B_1^{\top} P_2 B_1 & (D_1^u)^{\top}\\
        D_1^u & 0
\end{bmatrix}\right) \\
&= \begin{bmatrix}
        \Lambda_1^{uu} & \\
        \Lambda_1^{\xi u} & \Lambda_1^{\xi \xi}
\end{bmatrix} \\
\begin{bmatrix}
        K_1^u & K_1^{\xi}
\end{bmatrix} &= \begin{bmatrix}
        S_1^{\top} + A_1^{\top} P_2 B_1 & (D_1^x)^{\top}
\end{bmatrix} \Lambda_1^{-\top} \\
J_1 &= \operatorname*{blockdiag}(I_{n_u}, -I_{\rho_1}) \,.
\end{aligned}
\]
Following the procedure above we can begin to fill in the LDL of $\bar{\mathcal{K}}$ as follows
\[
\begin{aligned}
        \bar{L} = \begin{bmatrix}
                \Lambda_1^{uu} & & \hphantom{\Lambda_0^{uu}} & \hphantom{\Lambda_0^{\xi\xi}} \\
                \Lambda_1^{\xi u} & \Lambda_1^{\xi \xi} & & \\
                B_0^{\top} K_1^u & - B_0^{\top} K_1^{\xi} & \multicolumn{2}{c}{\multirow{2}{*}{$\star$}}\\
                0 & 0 & &
        \end{bmatrix} \\
        \Sigma = \operatorname*{blockdiag}(I_{n_u}, -I_{\rho_1}, \star) \,.
\end{aligned}
\]
We then proceed to update the $(0,0)$-block of $\bar{\mathcal{K}}$. Notice
that only its top-left sub-block will be updated

\begin{gather*}
\bar{\mathcal{K}}_{0,0} \leftarrow 
\begin{bmatrix}
        R_0 + B_0^{\top} \bigl(Q_1 + A_1^{\top} P_2 A_1 - K_1^u (K_1^u)^{\top} + K_1^{\xi} (K_1^{\xi})^{\top}\bigr) B_0 & (D_0^u)^{\top} \\
        D_0^u & 0
\end{bmatrix} = \begin{bmatrix}
        R_0 + B_0^{\top} P_1 B_0 & (D_0^u)^{\top} \\
        D_0^u & 0
\end{bmatrix}
\end{gather*}
where we defined 
\[
    P_1 = Q_1 + A_1^{\top}P_2 A_1 - K_1^u (K_1^u)^{\top} + K_1^{\xi} (K_1^{\xi})^{\top}\,.
\]
This is the crucial step: the Schur complement update propagates the cost-to-go
matrix $P_t$, obtained from the Riccati recursion stated in
Equation~\eqref{eq:eqcon_lqr_riccati}. Then defining 
\[
\begin{aligned}
        \Lambda_0 &= \begin{bmatrix}
                \Lambda_0^{uu}    & \\
                \Lambda_0^{\xi u} & \Lambda_0^{\xi \xi}
        \end{bmatrix} = \operatorname*{signedchol}\left(\begin{bmatrix}
                R_0 + B_0^{\top} P_1 B_0 & (D_0^u)^{\top} \\
                D_0^u                    & 0
        \end{bmatrix}\right) \\
        J_0 &= \operatorname*{blockdiag}(I_{n_u}, -I_{\rho_0})
\end{aligned}
\]
the complete LDL factorization of $\bar{\mathcal{K}}$ is given by
\begin{equation}
\label{eq:N2_LDL}
\begin{aligned}
        \bar{L} &= \begin{bmatrix}
                \Lambda_1^{uu} & & & \\
                \Lambda_1^{\xi u} & \Lambda_1^{\xi \xi} & & \\
                B_0^{\top} K_1^u & -B_0^{\top} K_1^{\xi} & \Lambda_0^{uu} & \\
                0 & 0 & \Lambda_0^{\xi u} & \Lambda_0^{\xi \xi}
        \end{bmatrix} \\
        \Sigma &= \operatorname*{blockdiag}(I_{n_u}, -I_{\rho_1}, I_{n_u}, -I_{\rho_0}) \,.
\end{aligned}
\end{equation}
We now move on to show that the LDL in~\eqref{eq:N2_LDL} allows us to solve the linear systems of Theorem~\ref{thm:MuMxi_rows_computation} through the stated recursion. Suppose that the newly added constraint is associated to stage index $t=1$ (we show the $t=1$ case because it is more illustrative than $t=0$ or $t=2$, but the same procedure applies in the other two cases), and call $c_u$ and $c_x$ the corresponding rows of $C_1^u$ and $C_1^x$, respectively. Let $g$ be the vector right-hand side defined by Equation~\eqref{eq:MuMxi_linear_system}. We reorder the slices of $g$ according to the proposed variable ordering, and aim to solve the linear system
\[
\begin{bmatrix}
        \Lambda_1^{uu} & & & \\
        \Lambda_1^{\xi u} & \Lambda_1^{\xi \xi} & & \\
        B_0^{\top} K_1^u & -B_0^{\top} K_1^{\xi} & \Lambda_0^{uu} & \\
        0 & 0 & \Lambda_0^{\xi u} & \Lambda_0^{\xi \xi}
\end{bmatrix} \begin{bmatrix}
        y_1 \\
        v_1 \\
        y_0 \\
        v_0
\end{bmatrix} = \begin{bmatrix}
        c_u \\
        0 \\
        B_0^{\top} c_x \\
        0
\end{bmatrix} \,.
\]
The solution of the system above is given by
\[
\begin{aligned}
        \begin{bmatrix}
                y_1 \\
                v_1
        \end{bmatrix} &= \Lambda_1^{-1} \begin{bmatrix}
                c_u \\
                0
        \end{bmatrix} \\
        \begin{bmatrix}
                y_0 \\
                v_0
        \end{bmatrix} &= \Lambda_0^{-1} \begin{bmatrix}
                B_0^{\top}p_1 \\
                0
        \end{bmatrix} \\
        p_1 &= c_x - K_1^u y_1 + K_1^{\xi} v_1
\end{aligned}
\]
which is exactly the stated recursion.

Finally, to conclude the proof we need to show that $\bar{L}$ can be used to construct, in the original basis, a factorization of $\mathcal{K}$ that reads
\[
\mathcal{K} = L \operatorname*{blockdiag}(I_n, -I_{\rho}) L^{\top} \,.
\]
This is easily obtained by setting $L = \Pi^{\top} \bar{L} \Pi$ where $\Pi$ is the permutation matrix associated to the proposed variable reordering.
\end{proof}
The proof for the $N=2$ case highlights the key fact that we rely on to
construct a factorization of the KKT matrix that is related to the Riccati
recursion. After reordering the rows and columns of $\mathcal{K}$, when running
the typical LDL factorization procedure, Schur complement updates propagate the
cost-to-go matrices given by the Riccati recursion. For this reason, the
obtained LDL has a specific structure that can be used when solving linear
systems involving the lower triangular factor.

We now move on to prove the general $N \geq 1$ case through an inductive
argument. As before, we first show that after variable reordering we can obtain
an LDL of $\bar{\mathcal{K}}$ with specific structure, and then use this to
prove that the linear system in Theorem~\ref{thm:MuMxi_rows_computation} can be
solved through the stated recursion.

\begin{proof}[Proof of Theorem~\ref{thm:MuMxi_rows_computation}]
Let $\bar{\mathcal{K}}$ be the KKT matrix of
Equation~\eqref{eq:KKT_matrix}, after row-column reordering given by
\[
u_{N-1},\,\xi_{N-1},\,u_{N-2},\,\xi_{N-2},\,\dots,\,u_0,\xi_0 \,.
\]
In the following we are going to refer to matrix products of the form $A_{\tau}
\dots A_{s}$, where we use the convention $A_{\tau} \dots A_{s} = I$ if $\tau <
s$. Similarly, for the transposed version $A_{s}^{\top} \dots A_{\tau}^{\top} =
I$, if $\tau < s$.

We rely on an inductive argument based on the hypothesis that, while applying
the LDL factorization procedure, before processing the $(k,k)$-block along the
diagonal, the matrix $\bar{L}$ is such that
\begin{equation}
\label{eq:L_structure}
\begin{aligned}
\bar{L}_{\tau,\tau} &= \Lambda_\tau\\
\bar{L}_{s,\tau} &= \begin{bmatrix}
        B_{s}^{\top}A_{s+1}^{\top} \cdots A_{\tau-1}^{\top} K_{\tau}^u & - B_s^{\top} A_{s+1}^{\top}\cdots A_{\tau-1}^{\top} K_{\tau}^{\xi}\\
        0 & 0
\end{bmatrix}
\end{aligned}
\end{equation}
for $\tau > k,\,s<\tau$. Moreover $\Sigma_{\tau} = J_{\tau}$ for $\tau > k$, where 
\begin{equation}
\label{eq:J_definition}
J_{\tau} = \operatorname*{blockdiag}(I_{n_u}, -I_{\rho_{\tau}}) \,.
\end{equation}
Finally, with a slight abuse of notation, denoting by $\bar{\mathcal{K}}$
the part of the same matrix that we have yet to factorize, after applying the
Schur complement updates from the LDL procedure, we assume that for $\tau \leq
k$ and $s < \tau$
\[
\begin{aligned}
    \bar{\mathcal{K}}_{\tau, \tau} &= \begin{bmatrix}
        \bar{H}_{\tau, \tau} & (D_{\tau}^u)^{\top} \\
        D_{\tau}^u & 0
\end{bmatrix} \\
    \bar{\mathcal{K}}_{s, \tau} &= \begin{bmatrix}
        M_{s,\tau} & N_{s,\tau} \\
        0 & 0
\end{bmatrix}
\end{aligned}
\]
where the first block of $\bar{\mathcal{K}}_{\tau, \tau}$ is given by
\[
\begin{aligned}
    \bar{H}_{\tau, \tau} &= R_{\tau} + \sum_{j=\tau+1}^{k} B_{\tau}^{\top} A_{\tau+1}^{\top} \cdots A_{j-1}^{\top} Q_j A_{j-1} \cdots A_{\tau+1} B_{\tau} \\
                             &+ B_{\tau}^{\top} A_{\tau+1}^{\top}\cdots A_{k}^{\top} P_{k+1} A_{k} \cdots A_{\tau+1} B_{\tau} \\
\end{aligned}
\]
and the two non-zero blocks of $\bar{\mathcal{K}}_{s, \tau}$ are given by
\[
\begin{aligned}
                  M_{s,\tau} &= B_s^{\top} A_{s+1}^{\top} \cdots A_{\tau-1}^{\top} S_{\tau}^{\top} +\sum_{j=\tau+1}^{k} B_s^{\top} A_{s+1}^{\top} \cdots A_{j-1}^{\top} Q_j A_{j-1} \cdots A_{\tau+1} B_{\tau} \\
                             &+B_s^{\top} A_{s+1}^{\top} \cdots A_{k}^{\top} P_{k+1} A_k \cdots A_{\tau+1} B_{\tau} \\
\end{aligned}
\]
and
\[
\begin{aligned}
                  N_{s,\tau} &= B_s^{\top} A_{s+1}^{\top} \cdots A_{\tau-1}^{\top} (D_{\tau}^x)^{\top} \,.
\end{aligned}
\]
Here $P_t$, $\Lambda_t$, $K_t^u$ and $K_t^{\xi}$ are defined as
in~\eqref{eq:eqcon_lqr_riccati},~\eqref{eq:eqcon_lqr_Lambda},~\eqref{eq:eqcon_lqr_K}.
Notice that the hypothesis is verified when $k=N-1$, which constitutes our base
case. 

Now we proceed with the next step in the factorization procedure. By our
inductive hypothesis, we have
\[
\bar{\mathcal{K}}_{k,k} = \begin{bmatrix}
        R_k + B_k^{\top} P_{k+1} B_k & (D_k^u)^{\top} \\
        D_k^u & 0
\end{bmatrix}
\]
and thus, from the definition of $\Lambda_k$ and $J_k$
\[
\bar{\mathcal{K}}_{k,k} = \Lambda_k J_k \Lambda_k^{\top}  \,.
\]
This implies that
$\bar{L}_{k,k} = \Lambda_k$, $\Sigma_k = J_k$. Then from the assumption on
$\bar{\mathcal{K}}_{s,k}$ for $s<k$, we have
\[
    \bar{L}_{s,k} = \begin{bmatrix}
            M_{s,k} & N_{s,k} \\
            0 & 0
    \end{bmatrix} \bar{L}_{k,k}^{-\top} \Sigma_k^{-1} = \begin{bmatrix}
            B_s^{\top} A_{s+1}^{\top} \cdots A_{k-1}^{\top} K_k^u & - B_s^{\top} A_{s+1}^{\top} \cdots A_{k-1}^{\top} K_k^{\xi} \\
            0 & 0
    \end{bmatrix} \,.
\]
Notice that the matrices $\bar{L}_{s,k}$ for $s\leq k$ have exactly the
expression stated by our inductive hypothesis in
Equation~\eqref{eq:L_structure}.

We now update the trailing Schur complement of $\bar{\mathcal{K}}$. Due to
the sparsity pattern of $\bar{L}_{s,k}$, only the $\bar{H}_{\tau,\tau}$ and
$M_{s,\tau}$ blocks are updated. In particular we have
\[
\begin{aligned}
        \bar{H}_{\tau,\tau} &\leftarrow R_{\tau} + \sum_{j=\tau+1}^{k} B_{\tau}^{\top} A_{\tau+1}^{\top} \cdots A_{j-1}^{\top} Q_j A_{j-1} \cdots A_{\tau+1} B_{\tau} \\
                            % &+ B_{\tau}^{\top} A_{\tau+1}^{\top}\cdots A_{k}^{\top} P_{k+1} A_{k} \cdots A_{\tau+1} B_{\tau} \\
                            & + B_{\tau}^{\top} A_{\tau+1}^{\top} \cdots A_{k-1}^{\top} \left(A_{k}^{\top} P_{k+1} A_{k} - K_k^u (K_k^u)^{\top} + K_k^{\xi}(K_k^{\xi})^{\top}\right) A_{k-1} \cdots A_{\tau+1} B_{\tau} \\
                            % & + B_{\tau}^{\top} A_{\tau+1}^{\top} \cdots A_{k-1}^{\top} K_k^{\xi}(K_k^{\xi})^{\top} A_{k-1} \cdots A_{\tau+1} B_{\tau}\\
                            &= R_{\tau} + \sum_{j=\tau+1}^{k-1} B_{\tau}^{\top} A_{\tau+1}^{\top} \cdots A_{j-1}^{\top} Q_j A_{j-1} \cdots A_{\tau+1} B_{\tau} \\
                            &+ B_{\tau}^{\top} A_{\tau+1}^{\top}\cdots A_{k-1}^{\top} P_{k} A_{k-1} \cdots A_{\tau+1} B_{\tau} \\
\end{aligned}
\]
where we used the fact that $P_k = Q_k + A_k^{\top}P_{k+1}A_k - K_k^u
(K_k^u)^{\top} + K_k^{\xi}(K_k^{\xi})^{\top}$. The expression for
$\bar{H}_{\tau,\tau}$ after the update also agrees with our
inductive hypothesis. A similar computation shows that the same is true for
$M_{s,\tau}$, $s<\tau< k$. This proves that the inductive hypothesis continues
to hold after processing the $(k,k)$-block along the diagonal, and consequently
we obtain the desired factorization, given by Equation~\eqref{eq:L_structure}
and $\tau=0,1,\dots,N-1$, $s<\tau$.

We now proceed to show that we can use this factorization to solve the linear
system of Theorem~\ref{thm:MuMxi_rows_computation} through the recursion
in~\eqref{eq:dual_recursion},~\eqref{eq:thm1_backward_rec}. Let $t\leq N-1$ (we address the $t=N$ case later) be
the time index associated to the constraint entering the active set, and denote
by $c_u$ and $c_x$ the corresponding rows of $C_t^u$, $C_t^x$. Denote by $g$
the right-hand side of Equation~\eqref{eq:MuMxi_linear_system}, and by
$\bar{g}$ the same vector after reordering according to 
\[
u_{N-1},\,\xi_{N-1},\,u_{N-2},\,\xi_{N-2},\,\dots,\,u_0,\,\xi_0 \,.
\]
We define the $\tau$-slice $\bar{g}_{\tau}$ as the vector obtained from the
entries corresponding to $(u_{\tau}$, $\xi_{\tau})$ of $\bar{g}$. Notice
that
\[
\begin{aligned}
\bar{g}_{\tau} &= 0\,,\quad \tau > t \\
\bar{g}_{t} &= \begin{bmatrix}
        c_u \\
        0
\end{bmatrix} \\
\bar{g}_{\tau} &= \begin{bmatrix}
        B_{\tau}^{\top} A_{\tau+1}^{\top} \cdots A_{t-1}^{\top} c_x \\ 0
\end{bmatrix}\,,\quad \tau < t \,.
\end{aligned}
\]
Denote also by $\bar{m}$ the vector $\begin{bmatrix}
        m_u^{\top} & m_{\xi}^{\top}
\end{bmatrix}^{\top}$ after the same reordering, whose slices are $y_t$ and
$v_t$. Solving the system
\[
\bar{L} \bar{m} = \bar{g}
\]
we see that for $\tau > t$, $y_{\tau} = 0$ and $v_{\tau} = 0$. Moreover
\[
\begin{bmatrix}
        y_t \\
        v_t 
\end{bmatrix} = \Lambda_t^{-1} \begin{bmatrix}
        c_u \\
        0
\end{bmatrix} \,.
\]
We can then update the right-hand side of the linear system with the
contribution from $y_t$ and $v_t$. From the structure of $\bar{L}_{\tau, t}$ in
Equation~\eqref{eq:L_structure} we get
\[
\begin{aligned}
\bar{g}_{\tau} &\leftarrow \bar{g}_{\tau} - \begin{bmatrix} B_{\tau}^{\top} A_{\tau+1}^{\top} \cdots A_{t-1}^{\top} \left(K_t^u y_t - K_t^{\xi}v_t\right) \\0 \end{bmatrix} \\
&= \begin{bmatrix} B_{\tau}^{\top} A_{\tau+1}^{\top} \cdots A_{t-1}^{\top}\left(c_x - K_t^u y_t + K_t^{\xi}v_t\right) \\ 0 \end{bmatrix} \\
&= \begin{bmatrix} B_{\tau}^{\top} A_{\tau+1}^{\top} \cdots A_{t-1}^{\top} p_t \\ 0 \end{bmatrix} \,.
\end{aligned}
\]
where we used the definition of $p_t$ from Equation~\eqref{eq:dual_recursion}.
Now assume that while solving the linear system, before solving for $y_k$ and
$v_k$, the remainder of the right-hand side, after applying the updates from
$y_{\tau}$ and $v_{\tau}$ for $\tau > k$, reads
\[
\bar{g}_{\tau} = \begin{bmatrix}
        B_{\tau}^{\top}A_{\tau+1}^{\top}\cdots A_{k}^{\top} p_{k+1} \\
        0
\end{bmatrix} \,.
\]
Notice that this assumption holds for the $k=t-1$ case. Then we can solve for
$y_k$, $v_k$ obtaining
\[
\begin{bmatrix}
        y_k \\
        v_k
\end{bmatrix} = \Lambda_k^{-1} \begin{bmatrix}
        B_{k}^{\top} p_{k+1} \\
        0
\end{bmatrix}
\]
and then update the remainder of the right hand side, which gives for $\tau < k$
\[
\begin{aligned}
        \bar{g}_{\tau} &\leftarrow \bar{g}_{\tau} - \begin{bmatrix} B_{\tau}^{\top} A_{\tau+1}^{\top} \cdots A_{k-1}^{\top}\left(K_k^{u}y_k - K_k^{\xi}v_k\right)\\ 0 \end{bmatrix} \\
        &= \begin{bmatrix} B_{\tau}^{\top} A_{\tau+1}^{\top} \cdots A_{k-1}^{\top} \left(A_k^{\top} p_{k+1} - K_k^u y_k + K_k^{\xi}v_k\right) \\ 0 \end{bmatrix} \\
        &= \begin{bmatrix} B_{\tau}^{\top} A_{\tau+1}^{\top} \cdots A_{k-1}^{\top} p_k \\ 0 \end{bmatrix} \,.
\end{aligned}
\]
Therefore the structure of the right-hand side is preserved after the update,
and by induction we have that $y_k$ and $v_k$ are indeed given by the recursion
in~\eqref{eq:dual_recursion} and~\eqref{eq:thm1_backward_rec}. Notice that the
inductive argument above applies also to the $t=N$ case, where $g_N =
\begin{bmatrix}
        p_{N}^{\top} B_{N-1} & 0
\end{bmatrix}^{\top}$, with $p_N = c_x$.

Finally, the LDL factorization $\bar{\mathcal{K}} = \bar{L} \Sigma
\bar{L}^{\top}$ can be used to construct $L = \Pi^{\top} \bar{L} \Pi$, where
$\Pi$ is the permutation matrix associated with the proposed variable
reordering. The matrix $L$ satisfies $\mathcal{K} = L
\operatorname*{blockdiag}(I_n, -I_{\rho})L^{\top}$, and thus the proof is
concluded. 
\end{proof}

\subsection{Equality constrained LQR problems}
\label{app:eqcon_lqr}
Consider a generic LQR problem with stagewise equality constraints, under the
cost assumptions of Problem~\eqref{prob:ocp_problem}, of the form
\begin{equation}
\label{prob:eqcon_lqr}
\begin{array}{ll}
    \displaystyle\operatorname*{minimize}_{\{x_t\},\{u_t\}} & \displaystyle \sum_{t=0}^{N-1}
    \left(
    \frac{1}{2}
    \begin{bmatrix}
            u_t\\x_t
    \end{bmatrix}^{\top}
    \begin{bmatrix}
            R_t & S_t \\
            S_t^{\top} & Q_t
    \end{bmatrix}
    \begin{bmatrix}u_t\\x_t\end{bmatrix}
    +
    \begin{bmatrix}r_t\\q_t\end{bmatrix}^{\top}
    \begin{bmatrix}u_t\\x_t\end{bmatrix}
    \right) + \frac{1}{2} x_N^{\top} Q_N x_N +q_N^{\top} x_N\\
    \text{subject to} &x_{t+1}=A_t x_t + B_t u_t + w_t, \qquad t=0,\ldots,N-1,\\
                      &D_t^u u_t + D_t^x x_t = d_t, \qquad t=0,\ldots,N-1,\\
                      &D_Nx_N=d_N,\\
                      &x_0\text{ given}.
\end{array}
\end{equation}
We are going to show that the problem above can always be rewritten as
\begin{equation}
\label{prob:eqcon_lqr_processed}
\begin{array}{ll}
    \underset{\{x_t\},\{u_t\}}{\text{minimize}} & J_{\text{LQR}}\\
    \text{subject to} &x_{t+1} = A_t x_t + B_t u_t + w_t, \qquad t=0,\ldots,N-1,\\
                      &\bar{D}_t^u u_t+ \bar{D}_t^x x_t= \bar{d}_t, \qquad t=0,\ldots,N-1,\\
                      &N_0x_0=n_0,\\
                      &x_0\text{ given},
\end{array}
\end{equation}
where all the matrices $\bar{D}_t^u$ are full row rank. The equality $N_0x_0=n_0$, which was not present explicitly
in the original formulation, establishes whether the problem is solvable for
the given initial state.

We employ a Dynamic Programming argument, and begin to solve the problem by
optimizing over $u_{N-1}$. After defining $P_N=Q_N$, $p_N=q_N$, we solve
\begin{equation}
\label{prob:last-stage}
\begin{array}{ll}
    \underset{u_{N-1},x_N}{\text{minimize}} & \displaystyle \frac12
    \begin{bmatrix}
            u_{N-1}\\
            x_{N-1}
    \end{bmatrix}^{\top}
    \begin{bmatrix}
            R_{N-1}&S_{N-1}\\
            S_{N-1}^{\top}&Q_{N-1}
    \end{bmatrix}
    \begin{bmatrix}
            u_{N-1}\\x_{N-1}
    \end{bmatrix} +
    \begin{bmatrix}
            r_{N-1}\\q_{N-1}
    \end{bmatrix}^{\top}
    \begin{bmatrix}
            u_{N-1}\\x_{N-1}
    \end{bmatrix}
    +\frac12x_N^{\top} P_N x_N + p_N^{\top} x_N\\
    \text{subject to} \quad & x_N = A_{N-1} x_{N-1} + B_{N-1} u_{N-1} + w_{N-1},\\
    &D_{N-1}^u u_{N-1} + D_{N-1}^x x_{N-1} = d_{N-1},\\
    &D_N x_N = d_N.
\end{array}
\end{equation}
We then substitute the expression for $x_N$ given by the dynamics into the
objective and the constraints. For the constraints this gives
\begin{equation}
\label{eq:last-constraint-aggregation}
\begin{bmatrix}
 D_{N-1}^u & D_{N-1}^x\\
D_N B_{N-1} & D_N A_{N-1}
\end{bmatrix}
\begin{bmatrix}
        u_{N-1}\\
        x_{N-1}
\end{bmatrix}
=
\begin{bmatrix}
d_{N-1}
\\
d_N - D_N w_{N-1}
\end{bmatrix}.
\end{equation}
We call the block column multiplying $u_{N-1}$
\begin{equation*}
F_{N-1} =
\begin{bmatrix}D_{N-1}^u \\ D_N B_{N-1} \end{bmatrix}.
\end{equation*}
Performing a rank-revealing transformation on $F_{N-1}$, for example by
Gaussian elimination with pivoting, we obtain an invertible matrix $J_{N-1}$
such that
\begin{equation*}
J_{N-1} F_{N-1} = \begin{bmatrix}
        \bar{D}_{N-1}^u \\ 0
\end{bmatrix}
\end{equation*}
where $\bar{D}_{N-1}^u$ is full row-rank. Premultiplying both sides of
\eqref{eq:last-constraint-aggregation} by this same matrix $J_{N-1}$ gives
\begin{equation}
\label{eq:last-rank-revealing}
\begin{bmatrix}
\bar{D}_{N-1}^u & \bar{D}_{N-1}^x \\
0 & N_{N-1}
\end{bmatrix}
\begin{bmatrix} u_{N-1} \\ x_{N-1} \end{bmatrix}
=
\begin{bmatrix}
\bar{d}_{N-1} \\ n_{N-1}
\end{bmatrix},
\end{equation}

This transformation of the equality constraints can be interpreted as follows:
given the current value of $x_{N-1}$, by choosing $u_{N-1}$ appropriately we
can satisfy only a subset of the constraints imposed by
\begin{equation*}
D_{N-1}^u u_{N-1} + D_{N-1}^x x_{N-1} = d_{N-1},
\qquad D_N x_N = d_N.
\end{equation*}
Our ability to satisfy all the constraints therefore becomes a condition on
$x_{N-1}$ itself, namely
\begin{equation}\label{eq:new-terminal-condition}
N_{N-1} x_{N-1} = n_{N-1}.
\end{equation}
At this point, Problem~\eqref{prob:last-stage} can be solved by substituting
the dynamics in the cost and solving an appropriate KKT system of the form
\begin{equation}\label{eq:last-stage-kkt}
    \begin{bmatrix}
    R_{N-1} + B_{N-1}^{\top} P_N B_{N-1} & (\bar{D}_{N-1}^u)^{\top}\\
    \bar{D}_{N-1}^u & 0
    \end{bmatrix}
    \begin{bmatrix}u_{N-1}\\\xi_{N-1}\end{bmatrix}
    =
    \begin{bmatrix}
        -g_{N-1}(x_{N-1}) \\ b_{N-1}(x_{N-1}) 
    \end{bmatrix}.
\end{equation}
where the right-hand side is an affine function of $x_{N-1}$. The solution is unique because
the Hessian of the KKT system is positive definite and $\bar{D}_{N-1}^u$ has full
row rank. Moreover, the optimal value is quadratic in $x_{N-1}$ and has
the form
\begin{equation}\label{eq:terminal-cost-to-go}
\frac12 x_{N-1}^{\top} P_{N-1} x_{N-1} + p_{N-1}^{\top} x_{N-1} + \text{const}.
\end{equation}
It can be verified that $P_{N-1}$ has the expression given by
Equation~\eqref{eq:eqcon_lqr_riccati} (after replacing $D_t^u$ and $D_t^x$ with
$\bar{D}_t^u$ and $\bar{D}_t^x$), and that $p_{N-1}$ is obtained from a
recursion similar to~\eqref{eq:thm1_backward_rec}. However, for the purpose of
this argument we don't need their exact value.

As in usual dynamic-programming fashion, after eliminating $u_{N-1}$ and $x_N$,
we obtain a problem in $x_1,\ldots,x_{N-1}$ and $u_0,\ldots,u_{N-2}$ with the
same structure, whose terminal cost is \eqref{eq:terminal-cost-to-go} and whose
terminal constraint is \eqref{eq:new-terminal-condition}.  The same procedure
can be applied recursively, with the propagated state constraint at stage $t$
denoted by $N_tx_t=n_t$. At the initial stage this becomes $N_0 x_0 = n_0$,
which is a constraint on the initial datum that determines whether the equality
constraints are feasible.

Notice that the constraint-transformation argument does not require solving for
$u_t$ and $x_{t+1}$ at the same time, and it can therefore be interpreted as a
pre-processing step. In practice we transform the constraints while solving the
Riccati to obtain $\Lambda_t$~\eqref{eq:eqcon_lqr_Lambda}, $K_t^u$ and
$K_t^{\xi}$~\eqref{eq:eqcon_lqr_K}.

\bibliographystyle{IEEEtran}
%\addtolength{\textheight}{-12cm} 
\begingroup
\renewcommand{\footnotesize}{\fontsize{8}{8}\selectfont}
\bibliography{references}
\endgroup

\end{document}